\documentclass[11pt, reqno]{amsart}
\usepackage[margin=1in]{geometry}
\usepackage{amsmath, amssymb, amsthm,bbm,bm}
\usepackage[shortlabels]{enumitem}
\usepackage{url}
\usepackage{tikz}
\usepackage[hidelinks]{hyperref}
\usepackage[capitalize]{cleveref}
\usepackage[noadjust]{cite}

\theoremstyle{plain}
\newtheorem{theorem}{Theorem}
\newtheorem{lemma}[theorem]{Lemma}
\newtheorem{corollary}[theorem]{Corollary} 
\newtheorem{proposition}[theorem]{Proposition} 
\newtheorem{claim}[theorem]{Claim}

\theoremstyle{definition}
\newtheorem{definition}[theorem]{Definition}

\theoremstyle{remark}
\newtheorem*{remark}{Remark}

\numberwithin{theorem}{section}
\numberwithin{equation}{section}

\newcommand{\floor}[1]{\left\lfloor #1 \right\rfloor}
\newcommand{\ceil}[1]{\left\lceil #1 \right\rceil}
\newcommand{\paren}[1]{\left( #1 \right)}

\newcommand{\sqb}[1]{\left[ #1 \right]}

\newcommand{\E}{\mathbb E}

\newcommand{\R}{\mathbb R}

\newcommand{\fB}{\mathcal B}
\newcommand{\fS}{\mathcal S}
\newcommand{\fX}{\mathcal X}

\newcommand{\cP}{\mathcal P}

\DeclareMathOperator{\vol}{vol}
\DeclareMathOperator{\inter}{int}

\title{Separators for intersection graphs of spheres} 
\author{Jacob Fox}
\author{Jonathan Tidor}
\address{Department of Mathematics, Stanford University, Stanford, CA 94305, USA}
\email{jacobfox@stanford.edu}
\address{Department of Mathematics, Princeton University, Princeton, NJ 08544, USA}
\email{jtidor@princeton.edu}

\thanks{Fox was supported by NSF Awards DMS-2452737 and DMS-2154129. Tidor was supported in part by a Stanford Science Fellowship. This material is based in part upon work supported by the National Science Foundation under Grant No. DMS-1928930, while the authors were in residence at the Simons Laufer Mathematical Sciences Institute in Berkeley, California, during the Spring 2025 semester.}
\date{}

\begin{document}
\begin{abstract}
We prove the existence of optimal separators for intersection graphs of balls and spheres in any dimension $d$. One of our results is that if an intersection graph of $n$ spheres in $\R^d$ has $m$ edges, then it contains a balanced separator of size $O_d(m^{1/d}n^{1-2/d})$. This bound is best possible in terms of the parameters involved. The same result holds if the balls and spheres are replaced by fat convex bodies and their boundaries.
\end{abstract}

\maketitle

\section{Introduction}
\label{sec:into}
For a graph $G=(V,E)$, a \emph{balanced separator} is a vertex set $\fX\subseteq V$ such that each connected component of $G\setminus \fX$ contains at most $\tfrac23|V|$ vertices. In 1979, Lipton and Tarjan \cite{LT79} showed that $n$-vertex planar graphs have balanced separators of size $O(\sqrt n)$. There are many different characterizations of planar graphs that each lead to a different extension of the Lipton--Tarjan separator theorem. For example, Wagner's theorem \cite{Wa37} states that a graph is planar if and only if it does not contain $K_5$ or $K_{3,3}$ as a minor. The family of graphs that can be embedded in a fixed surface has a finite list of forbidden minors by the Robertson--Seymour theorem \cite{RS04}. Gilbert, Hutchinson, and Tarjan \cite{GHT84} proved that every $n$-vertex graph which can be embedded in a surface of genus $g$ has a separator of size $O(\sqrt{gn})$. Later, Alon, Seymour, and Thomas \cite{AST90} proved that for every graph $H$, every $n$-vertex graph that does not contain $H$ as a minor has a separator of size at most $O_H(\sqrt{n})$. 

Another important characterization of planar graphs is the Koebe--Andreev--Thurston circle packing theorem, which says that a graph is planar if and only if it is the intersection graph of interior-disjoint disks in the plane. (This means that there is a collection of interior-disjoint disks, one corresponding to each vertex, so that two vertices are connected by an edge if and only if the corresponding disks are tangent.) This leads to another way to try to extend the Lipton--Tarjan separator theorem, to intersection graphs of geometric objects. A \emph{string graph} is an intersection graph of arcwise-connected sets in the plane, so planar graphs are string graphs. The Lipton--Tarjan separator theorem has been extended in several works studying separators for string graphs \cite{FP08,FP10,KoLo24,Lee17,Mat14}, and we now know that every string graph with $m$ edges has a balanced separator of size $O(\sqrt m)$. 

Graphs with a small balanced separator are poor expanders, and one can utilize the separator for many applications, including in fast divide-and-conquer type algorithms for many algorithmically hard graph problems (see, for example, \cite{LT80}), extremal applications (see, for example, \cite{FP14,LRT79}), and enumerative problems (see, for example, \cite{DN10}). 

In this paper, we are interested in finding small separators for intersection graphs of geometric objects in higher dimensions. Miller, Teng, Thurston, and Vavasis \cite{MTTV97} proved a separator theorem for the intersection graph of balls in  $\R^d$ in which no point is in more than $k$ of the balls. Generalizing this, we prove that every intersection graph of $n$ balls in $\R^d$ with $m$ edges has a balanced separator of size $O_d(m^{1/d}n^{1-2/d})$. This bound is best possible in terms of the parameters involved. (Here and throughout the paper, we use the notation $O_d(\cdot)$ to hide a multiplicative constant that depends only on the dimension $d$.)

Proving a separator theorem for intersection graphs of spheres turns out to be more challenging as such graphs can have significantly more complicated combinatorial structure. Recently, Davies, Georgakopoulos, Hatzel, and McCarty \cite{DGHM25} studied this problem under the constraint that the intersection graph is $K_{t,t}$-free, proving that if an intersection graph of $n$ spheres in $\R^d$ is $K_{t,t}$-free, then it has a balanced separator of size $\tilde{O}_d(t^{10}n^{1-1/(2d+8)})$. We first study the problem in general. Despite their more complicated structure, we show that intersection graphs of spheres have the same size separators as intersection graphs of balls.

\begin{theorem}
\label{thm:sphere-separator-main-intro}
An intersection graph of $n$ spheres in $\R^d$ with $m$ edges has a balanced separator of size $O_{d}(m^{1/d}n^{1-2/d})$.
\end{theorem}

As a corollary, we show that every $K_{t,t}$-free intersection graph of $n$ spheres has a separator of size $O_d(t^{1/d}n^{1-1/d})$ (see Corollary~\ref{cor:ktt-free-separator}). This verifies a strengthening of a conjecture of Davies, Georgakopoulos, Hatzel and McCarty \cite[Conjecture 20]{DGHM25}. Both Theorem~\ref{thm:sphere-separator-main-intro} and the corollary are best possible up to the constant factor depending on the dimension $d$ (see Proposition~\ref{prop:lower-bound}). We also prove a refinement of Theorem~\ref{thm:sphere-separator-main-intro} which gives a separator whose size depends on the degree sequence of the intersection graph. As a corollary, we prove that every $K_{t,t}$-free intersection graph of spheres in $\R^d$ has average degree $O_d(t)$. It follows that $K_{t,t}$-free intersection graphs of spheres in $\R^d$ have chromatic number $O_d(t)$. 

We further study intersection graphs of other geometric objects in Section~\ref{sec:further}. We generalize our results from balls and spheres to fat convex bodies and their boundaries. Unlike the planar case, we show that both the fatness and convexity assumptions are necessary in dimensions three and larger: no non-trivial separator exists for the intersection graph of general convex bodies or general fat bodies.

Finally, as separators have many algorithmic applications, in Section~\ref{sec:alg} we show that there are efficient randomized algorithms for finding such separators of intersection graphs of balls or spheres. 

\section{Separators for intersection graphs of balls}
Let $\fB$ be a collection of balls in $\R^d$. The \emph{intersection graph} of $\fB$, denoted $G(\fB)$, is the graph with vertex set $\fB$ where two balls are adjacent if they intersect. The collection $\fB$ is \emph{$k$-ply} if every point of $\R^d$ is contained in at most $k$ balls. Miller, Teng, Thurston, and Vavasis \cite{MTTV97} proved the following separator theorem for intersection graphs of balls.

\begin{theorem}[\cite{MTTV97}]
\label{thm:ball-separator-ply}
Let $\fB$ be a $k$-ply collection of $n$ balls in $\R^d$. Then $G(\fB)$ has a balanced separator of size $O_d(k^{1/d}n^{1-1/d})$.
\end{theorem}

A more recent proof of this result was given by Har-Peled \cite{Har11}. Using his techniques, we are able to give a more refined version of the Miller--Teng--Thurston--Vavasis separator theorem. (Similar techniques also appeared in an earlier work of Smith and Wormald \cite{SW98}.)

Given a collection of balls $\fB$, for each point $x\in\fB$, the \emph{ply} of $x$, denoted $p(x)$, is the number of balls in $\fB$ which contain $x$. For each ball $B\in\fB$, the \emph{ply} of $B$, denoted $p(B)$, is the maximum over $p(x)$ over all points $x\in B$. With this definition we strengthen Theorem~\ref{thm:ball-separator-ply} as follows.

\begin{theorem}
\label{thm:ball-separator-lp}
Let $\fB$ be a collection of $n$ balls in $\R^d$. Then $G(\fB)$ has a balanced separator of size 
\[O_d\paren{\paren{\sum_{B\in\fB}p(B)^{\tfrac1{d-1}}}^{1-\tfrac1d}}.\]
\end{theorem}

Since every ball has ply at most one more than its degree in the intersection graph, we have the following corollary which bounds the size of the separator in terms of the degree sequence of the graph, or simply its number of edges.

\begin{corollary}
\label{cor:ball-separator-edge}
Let $\fB$ be a collection of $n$ balls in $\R^d$. Then $G(\fB)$ has a balanced separator of size
\[O_d\paren{\paren{\sum_{B\in\fB}\deg(B)^{\tfrac1{d-1}}}^{1-\tfrac1d}}.\]
In particular, if $G(\fB)$ has $m$ edges, then it has a balanced separator of size $O_d(m^{1/d}n^{1-2/d})$.
\end{corollary}

The second claim follows from the first via an application of H\"older's inequality since
\[\sum_{B\in\fB}\deg(B)^{\tfrac1{d-1}}\leq\paren{\sum_{B\in\fB}\deg(B)}^{\tfrac1{d-1}}\paren{\sum_{B\in\fB}1}^{\tfrac{d-2}{d-1}}=(2m)^{\tfrac1{d-1}}n^{\tfrac{d-2}{d-1}}.\]

To construct a balanced separator, we first construct a separator with slightly worse quantitative properties and then iterate a bounded number of times to find the desired separator.

\begin{definition}
For a graph $G=(V,E)$ and $c<1$, a \emph{$c$-balanced separator} is a vertex set $\fX\subseteq V$ such that each connected component of $G\setminus \fX$ contains at most $c|V|$ vertices. 
\end{definition}

\begin{lemma}
\label{lem:balancing-separators}
Suppose graph $G=(V,E)$ is such that every induced subgraph has a $c$-balanced separator of size at most $s$. Then $G$ has a balanced separator of size at most $\lceil \log_c(2/3)\rceil s$.
\end{lemma}

\begin{proof}
Set $V_1=V$. Given a vertex set $V_i$, let $\fX_i$ be a $c$-balanced separator of $G[V_i]$ of size at most $s$. Define $V_{i+1}$ to be the set of vertices comprising the largest connected component of $G[V_i]\setminus\fX_i$. By definition, $|V_{i+1}|\leq c|V_i|$. Furthermore, each other connected components of $G[V_i]\setminus\fX_i$ contains at most $\tfrac 12|V_i|\leq\tfrac23|V|$ vertices. Halt the process at the first $i=r$ satisfying $|V_{r+1}|\leq\tfrac23|V|$, so $r \leq \lceil \log_c(2/3)\rceil$. Defining $\fX=\fX_1\cup\cdots\cup\fX_r$, we have $|\fX|\leq rs\leq \lceil \log_c(2/3)\rceil s$. Furthermore, every connected component of $G\setminus\fX$ is either equal to $V_{r+1}$, or is a connected component of $G[V_i]\setminus\fX_i$ other than $V_{i+1}$ for some $i$. Since all of these sets have size at most $\tfrac23|V|$, this implies that $\fX$ is a balanced separator for $G$.
\end{proof}

\begin{proof}[Proof of Theorem~\ref{thm:ball-separator-lp}]
We will give a randomized procedure that produces a $(1-5^{-d})$-balanced separator $\fX$ for $G(\fB)$ of the desired size. By Lemma~\ref{lem:balancing-separators}, losing another constant factor depending only on $d$, this suffices to prove the desired result.

Let $B_0$ be the ball in $\R^d$ with smallest radius that contains at least $5^{-d}n$ centers of balls of $\fB$. ($B_0$ is not necessarily an element of $\fB$.) By perturbing $B_0$ slightly, one can see that $B_0$ contains exactly $\ceil{5^{-d}n}$ centers. After translating and rescaling the problem appropriately, assume that $B_0=B(1)$, the origin-centered ball of radius 1.

For some $r\in[1,2]$, let $\fX$ be the set of balls $B\in\fB$ which intersect the origin-centered sphere of radius $r$. We claim that for any $r\in[1,2]$, the set $\fX$ is a $(1-5^{-d})$-balanced separator for $G(\fB)$. To see this, note that each connected component of $\fB\setminus \fX$ consists of balls which are either entirely inside $B(r)$ or entirely outside $B(r)$. Since $B(1)\subseteq B(r)$ contains $\ceil{5^{-d}n}$ centers, there are at most $(1-5^{-d})n$ balls entirely outside $B(r)$. Now $B(2)\supseteq B(r)$ can be covered by $4^d$ translates of $B(1)$ which, by the minimality of $B_0$, each contain at most $\ceil{5^{-d}n}$ centers. Thus $B(2)$ contains at most $4^d\ceil{5^{-d}n}\leq (1-5^{-d})n$ centers, so there are at most this many balls entirely inside $B(r)$.

We now show that there is a choice of $r\in[1,2]$ so that $\fX$ is small. To do this, pick $r\in[1,2]$ uniformly at random. We now bound the expected size of $\fX$. We do this by modifying $\fB$ by shrinking those balls that intersect $B(2)$ so that they lie entirely in $B(2)$. For each $B\in\fB$ which intersects $B(2)$, let $\ell_B$ be the line that passes through the origin and the center of $B$. Let $B'$ be the ball contained in $B(2)$ whose center lies on $\ell_B$ and which satisfies $\ell_B\cap B'=\ell_B\cap B\cap B(2)$. (See Figure~\ref{fig:constructing-B'}.) Define $\fB'$ to be the set of these balls $B'$ (the balls in $\fB$ that do not intersect $B(2)$ do not have any corresponding ball in $\fB'$).

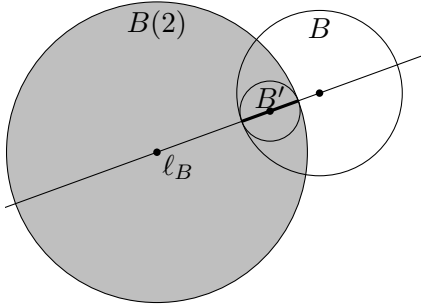
\begin{figure}
\centering
\begin{tikzpicture}[scale=1]
\clip (-2.01,-2.01) rectangle + (5.51,4.01);

\filldraw[fill=lightgray] (0,0) circle (2);

\draw(-3*0.9397,-3*0.342) -- (5*0.9397,5*0.342);
\draw[line width=1.2pt] (1.2*0.9397,1.2*0.342) -- (2*0.9397,2*0.342);

\draw(2.3*0.9397,2.3*0.342) circle (1.1);
\draw(1.6*0.9397,1.6*0.342) circle (0.4);

\node[] at (0.3*0.9397,0.3*0.342) [below] {$\ell_B$};
\node[] at (0,2+0.03) [below] {$B(2)$};
\node[] at (2.3*0.9397,2.3*0.342+1.1+0.03) [below] {$B$};
\node[] at (1.6*0.9397,1.6*0.342+0.4+0.07) [below] {$B'$};

\node[fill=black, circle, inner sep=1pt] at (0,0) {};
\node[fill=black, circle, inner sep=1pt] at (1.6*0.9397,1.6*0.342) {};
\node[fill=black, circle, inner sep=1pt] at (2.3*0.9397,2.3*0.342) {};
\end{tikzpicture}
\caption{\label{fig:constructing-B'}
Constructing $\fB'$ from $\fB$.}
\end{figure}

Since each ball $B'\in\fB'$ lies inside the corresponding ball $B\in \fB$, clearly we have $p_{\fB'}(B')\leq p_{\fB}(B)$. Write $p(x)$ for the number of balls of $\fB'$ which contain the point $x$. Then $p_{\fB'}(B)=\max_{x\in B}p(x)$. We now have
\begin{align*}
\vol(B(2))
&\geq \vol\paren{\bigcup_{B\in\fB'}B}
=\int_{\bigcup_{B\in\fB'}B}p(x)\frac{\mathrm dx}{p(x)}
=\sum_{B\in\fB'}\int_B \frac{\mathrm dx}{p(x)}\\
&\geq\sum_{B\in\fB'}\int_B \frac{\mathrm dx}{p_{\fB'}(B)}
=\sum_{B\in\fB'}\frac{\vol(B)}{p_{\fB'}(B)}.
\end{align*}
Thus, dividing by $\vol(B(1))$, we conclude that
\[2^d\geq\sum_{B\in\fB'}\frac{r(B)^d}{p_{\fB'}(B)}.\]

Now for each $B\in\fB$, the probability that $B$ is in the separator is at most the length of the interval $\ell_B\cap B\cap B(2)$. In particular, the probability is 0 if $B$ does not intersect $B(2)$. Now if $B$ does intersect $B(2)$, it corresponds to a ball $B'\in\fB'$ so that $\ell_B\cap B\cap B(2)=\ell_B\cap B'$ is an interval of length $2r(B')$. Thus by H\"older's inequality we conclude
\begin{align*}
\E[|\fX|]
&\leq 2\sum_{B\in\fB'} r(B)
\leq 2 \paren{\sum_{B\in\fB'}\frac{r(B)^d}{p_{\fB'}(B)}}^{1/d}\cdot \paren{\sum_{B\in\fB'}p_{\fB'}(B)^{1/(d-1)}}^{(d-1)/d}\\
&\leq 4\paren{\sum_{B\in\fB'} p_{\fB'}(B)^{1/(d-1)}}^{(d-1)/d}
\leq 4\paren{\sum_{B\in\fB} p_{\fB}(B)^{1/(d-1)}}^{(d-1)/d}.\qedhere
\end{align*}
\end{proof}

\begin{remark}
\label{rem:degeneracy}
Clearly Theorem~\ref{thm:ball-separator-lp} implies Theorem~\ref{thm:ball-separator-ply}. We point out that even the weakest version of Corollary~\ref{cor:ball-separator-edge} is sufficient to imply Theorem~\ref{thm:ball-separator-ply}. To see this, let $\fB$ be a $k$-ply collection of balls in $\R^d$. Then the intersection graph $G(\fB)$ is $O_d(k)$-degenerate. (Recall that a graph is $D$-degenerate if there exist an ordering of the vertices so that each vertex is adjacent to at most $D$ vertices which appear earlier in the ordering.) To see this, order the balls by radius. One can check that any ball $B_j$ which intersects $B_i$ and has radius satisfying $r_j\geq r_i$ must occupy at least a $(1/3)^d$-fraction of $B'_i$, defined to be the dilate of $B_i$ by a factor of 3. The $k$-ply assumption shows that there are at most $3^dk$ such balls $B_j$. Thus $G(\fB)$ has at most $3^dkn$ edges so Theorem~\ref{thm:ball-separator-ply} follows from Corollary~\ref{cor:ball-separator-edge}.
\end{remark}

\section{Separators for intersection graphs of spheres}

The main difficulty in generalizing the results of the previous section to intersection graphs of spheres is that the notion of ply is not useful for sphere intersection graphs. In particular, one can nest spheres so that the corresponding collection of balls would have large ply yet the sphere intersection graph has few or no edges. Furthermore, there are dense graphs (in particular, complete bipartite graphs) that are intersection graphs of spheres but are $2$-ply in the sense that no point lies on more than two spheres. (To see this, take two nested families of spheres that are arranged so that every sphere from one family intersects every sphere from the other family.) This shows that no analogue of the argument in Remark~\ref{rem:degeneracy} applies to sphere intersection graphs. Despite these difficulties, we show that separators of the same size exist.

\begin{theorem}
\label{thm:sphere-separator-main}
Let $\fS$ be a collection of $n$ spheres in $\R^d$. Then $G(\fS)$ has a balanced separator of size at most
\[O_d\paren{\paren{\sum_{S\in\fS}\deg(S)^{\tfrac1{d-1}}}^{1-\tfrac1d}}.\]
\end{theorem}

As with Corollary~\ref{cor:ball-separator-edge}, this result implies Theorem~\ref{thm:sphere-separator-main-intro} by an application of H\"older's inequality. We start with an elementary lemma that is needed for the proof.

\begin{lemma}
\label{lem:nested-sphere-separator}
Let $\fS$ be a collection of spheres in $\R^d$. Suppose that $G(\fS)$ has maximum degree at most $\Delta$ and that there exists a point $p\in\R^d$ contained in the interior of at least $2s$ spheres. Then there exists a set $\fX\subseteq\fS$ with $|\fX|\leq\Delta$ such that every connected component of $G(\fS)\setminus\fX$ has size at most $|\fS|-s$.
\end{lemma}

\begin{proof}
Let $\fS'\subseteq\fS$ be the set of spheres containing $p$ in their interior. For any two spheres $S,S'\in\fS'$, either they intersect or one contains the other in its interior. Consider the partial order on $\fS'$ defined by containment. Take any linear extension of this poset and let $S_0\in\fS'$ be the median element of this linear extension. This means that there are at least $s-1$ spheres which either intersect $S_0$ or are contained in its interior and there are at least $s-1$ spheres which either intersect $S_0$ or contain it in their interior.

Now let $\fX=N(S_0)$, the neighborhood of $S_0$ in $G(\fS)$, i.e., the set of spheres which intersect $S_0$. We claim that every connected component of $G(\fS)\setminus\fX$ has size at most $|\fS|-s$. To see this, note that no sphere in $\fS\setminus\fX$ intersects $S_0$, so any connected component of $G(\fS)\setminus\fX$ consists of spheres either entirely inside $S_0$ or entirely outside $S_0$. However, we showed that there are at least $s$ spheres that are not of the first type ($S_0$ and the $s-1$ spheres which follow it in the linear extension) and at least $s$ spheres that are not of the second type. Thus each connected component of $G(\fS)\setminus\fX$ has size at most $|\fS|-s$.
\end{proof}

\begin{proof}[Proof of Theorem~\ref{thm:sphere-separator-main}]
Set \[\Sigma=C_d\paren{\sum_{S\in\fS} \deg(S)^{\tfrac1{d-1}}}^{1-\tfrac1d}.\] Choosing $c_d>0$ small in terms of $d$ and $C_d$ large in terms of $c_d$, we will construct a $(1-c_d)$-balanced separator for $G(\fS)$ of size at most $\Sigma$. By Lemma~\ref{lem:balancing-separators}, losing another constant factor depending only on $d$, this suffices to prove the desired result. Assume $\Sigma<n$, since the result is trivial otherwise.

Through the proof, if we have not already found the desired separator, we will define sets $\fX_i \subset \fS$ for $i=1,2,3,4$ with $|\fX_i| \leq \Sigma/4$ for each $i$ so that $\fX_1\cup \fX_2\cup \fX_3\cup \fX_4$ is the desired separator.

\vspace{5pt} \noindent \textbf{Step 1:}
First, define $\fX_1\subset \fS$ to be the $\Sigma/4$ vertices of highest degree in $G(\fS)$. Let $\Delta$ be the maximum degree of the resulting graph $G(\fS)\setminus\fX_1$. We claim that $\Delta\leq\Sigma/4$ and $\Delta\leq c_dn$. If the first does not hold, then $G(\fS)$ would have at least $\Sigma/4$ vertices of degree at least $\Sigma/4$, meaning that
\[\Sigma=C_d\paren{\sum_{S\in\fS} \deg(S)^{1/(d-1)}}^{1-1/d}\geq C_d\paren{(\Sigma/4)(\Sigma/4)^{1/(d-1)}}^{1-1/d}=\frac{C_d}4\Sigma.\]
Clearly this is impossible for $C_d>4$. Similarly, if the second fails then $G(\fX)$ would have at least $\Sigma/4$ vertices of degree at least $c_dn$, meaning that
\[\Sigma\geq C_d\paren{(\Sigma/4)(c_dn)^{1/(d-1)}}^{1-1/d}=\frac{C_dc_d^{1/d}}{4^{1-1/d}}\Sigma^{1-1/d}n^{1/d}>\frac{C_dc_d^{1/d}}{4^{1-1/d}}\Sigma,\]
where the final inequality holds since $\Sigma<n$. Again, this is impossible for $C_d$ sufficiently large in terms of $c_d$.

\vspace{5pt} \noindent \textbf{Step 2:} Write $\fS_1=\fS\setminus \fX_1$. The goal of this step is to pass to a subset $\fS_2\subseteq \fS_1$ such that every sphere $S\in\fS_2$ contains fewer than $4c_dn$ centers of spheres in $\fS_2$. Let $S_0\in\fS_1$ be a sphere of minimal radius that contains $4c_dn$ centers of spheres in $\fS_1$. (If no such sphere exists, taking $\fS_2=\fS_1$ completes this step.) Define $\fS_{\supseteq S_0}\subseteq \fS_1$ to be the set of spheres which contain $S_0$ entirely within their interior. First suppose that $|\fS_{\supseteq S_0}|\geq 2c_dn$. In this case we can apply Lemma~\ref{lem:nested-sphere-separator} to find a separator $\fX_2\subset\fS_1$ such that $|\fX_2|\leq\Delta\leq\Sigma/4$ and each connected component of $G(\fS_1)\setminus\fX_2$ has size at most $(1-c_d)n$. Then $\fX_1\cup\fX_2$ is a $(1-c_d)$-balanced separator for $G(\fS)$ of size at most $\Sigma/2$.

Thus we can assume that $|\fS_{\supseteq S_0}|< 2c_dn$. In this case, define $\fX_2=N(S_0)$, the neighborhood of $S_0$ in $G(\fS)$. Define $\fS_2\subseteq\fS_1$ to be the set of spheres which are entirely contained within $S_0$. Clearly $|\fX_2|\leq \Delta\leq\Sigma/4$.

\begin{claim}
\label{claim:sphere-separator-proof}
If $G(\fS)\setminus(\fX_1\cup \fX_2)$ has a connected component of size more than $(1-c_d)n$, then it is also a connected component of $G(\fS_2)$. 
\end{claim}
\begin{proof}
Every connected component of $G(\fS)\setminus (\fX_1\cup \fX_2)$ is either made up of spheres entirely inside $S_0$ or entirely outside $S_0$. The former are also connected components of $G(\fS_2)$ while the latter are not. There are at least $4c_dn$ centers of spheres of $\fS_1$ inside of $S_0$. Since at most $2c_dn$ of these correspond to spheres surrounding $S_0$ and at most $\Delta$ correspond to spheres which intersect $S_0$, we see that there are at least $2c_dn-\Delta\geq c_dn$ spheres inside $S_0$. Therefore any connected component outside of $S_0$ has size at most $(1-c_d)n$.
\end{proof}

We now construct a separator for $G(\fS_2)$. We know that every sphere in $\fS_2$ contains fewer than $4c_dn$ centers of spheres in $\fS_2$. We can also assume that $|\fS_2|\geq(1-c_d)n$, otherwise $\fX_1\cup\fX_2$ is already the desired $(1-c_d)$-balanced separator.

\vspace{5pt} \noindent \textbf{Step 3:} Let $B_{\mathrm{in}}$ be a ball of minimal radius which contains at least $4c_dn$ centers of spheres in $\fS_2$. Since we showed that no sphere in $\fS_2$ contains $4c_dn$ centers, we see that no sphere in $\fS_2$ contains $B_{\mathrm{in}}$ entirely in its interior. This means that each of the $4c_dn$ spheres whose centers lie in $B_{\mathrm{in}}$ also intersect $B_{\mathrm{in}}$.

Write $R$ for the radius of $B_{\mathrm{in}}$ and define $B_{\mathrm{out}}$ to be the ball with the same center and radius $2R$. By the minimality of $B_{\mathrm{in}}$, we see that $B_{\mathrm{out}}$ contains at most $4^d\ceil{4c_dn}$ centers. In particular, choosing $c_d$ sufficiently small, there are at least $c_dn$ sphere centers outside of $B_{\mathrm{out}}$. Let $S_{\mathrm{rand}}$ be a sphere with the same center as $B_{\mathrm{in}}$ and $B_{\mathrm{out}}$ whose radius is chosen uniformly at random in $[R,2R]$.

For a sphere $S\in\fS$, write $r(S)$ for the radius of $S$. Write $\tilde r(S)$ for the radius of the spherical cap formed by intersecting the ball bounded by $S$ with $S_{\mathrm{rand}}$. (The radius of a spherical cap is the minimal $\tilde r(S)$ so that the cap is contained in a ball of radius $\tilde r(S)$.)

Choose some $C_d'$ sufficiently large in terms of $d$ and let $\fX_3$ consist of the spheres $S\in \fS_2$ which intersect $S_{\mathrm{rand}}$ and which satisfy $\tilde r(S)\leq C_d' R(\deg(S)/\Sigma)^{1/(d-1)}$.

\begin{claim}
\label{claim:expected-x3-size}
$\E[|\fX_3|]\leq\Sigma/4$.
\end{claim}
\begin{proof}
We need to upper bound the probability that a sphere $S$ intersects $S_{\mathrm{rand}}$ in a spherical cap of radius at most $\tilde r$. If $\tilde r\leq r(S)$, the probability is upper bounded by 
\[\frac{2\paren{r(S)-\sqrt{r(S)^2-\tilde r^2}}}R\leq \frac{2\tilde r}{R}.\]
To see this, note that this event is the same as the event that $S_{\mathrm{rand}}$ intersects the line connecting the center of $B_{\mathrm{in}}$ and the center of $S$ in one of the bold segments in Figure~\ref{fig:expected-x3-size}.

Now if $\tilde r>r(S)$, the probability we need to upper bound is just the probability that $S_{\mathrm{rand}}$ and $S$ intersect. This quantity is at most $2r(S)/R<2\tilde r/R$. 

Since $\fX_3$ is defined with the cutoff $\tilde r = C_d' R(\deg(S)/\Sigma)^{1/(d-1)}$, for each $S\in\fS_2$ we have the bound
\[\Pr\left[S\text{ intersects }S_{\mathrm{rand}}\text{ and }\tilde r(S)\leq C_d' R(\deg(S)/\Sigma)^{1/(d-1)}\right]\leq 2C_d'(\deg(S)/\Sigma)^{1/(d-1)}.\]
This implies that
\begin{align*}
\E[|\fX_3|]
&\leq \sum_{S\in\fS_2} 2C_d'(\deg(S)/\Sigma)^{1/(d-1)}
=\frac{2C_d'}{\Sigma^{1/(d-1)}C_d^{d/(d-1)}}\Sigma^{d/(d-1)}
\leq \Sigma/4,
\end{align*}
where the equality is by the definition of $\Sigma$ and the last inequality holds for $C_d$ chosen appropriately large in terms of $C_d'$.
\end{proof}

Fix some choice of $S_{\mathrm{rand}}$ for which $|\fX_3|\leq\Sigma/4$. 

\begin{figure}
\centering
\begin{tikzpicture}[scale=4]
\clip (1.4-1,-0.5) rectangle + (2,1);

\draw(0,0) circle (1);
\draw(0,0) circle (2);

\draw(-2,0) -- (5,0);

\draw(1.4,0) circle (0.25);
\draw(1.4,0)--(1.4-0.6*0.25,0.8*0.25)--(1.4-0.6*0.25,0);
\draw(1.4,0)--(1.4+0.6*0.25,-0.8*0.25)--(1.4+0.6*0.25,0);
\draw[line width=2pt] (1.4-0.25,0)--(1.4-0.6*0.25,0);
\draw[line width=2pt] (1.4+0.25,0)--(1.4+0.6*0.25,0);

\node[fill=black, circle, inner sep=1pt] at (1.4,0) {};
\node[] at (1.4,-0.25) [below] {$S$};
\node[] at (1.4-0.6*0.25-0.03,0.8*0.25*0.5) [] {$\tilde r$};
\node[] at (1.4-0.6*0.25*0.5+0.09,0.8*0.25*0.5+0.03) [] {$r(S)$};

\end{tikzpicture}
\caption{\label{fig:expected-x3-size}
Bounding $\Pr[S\in\fX_3]$.}
\end{figure}
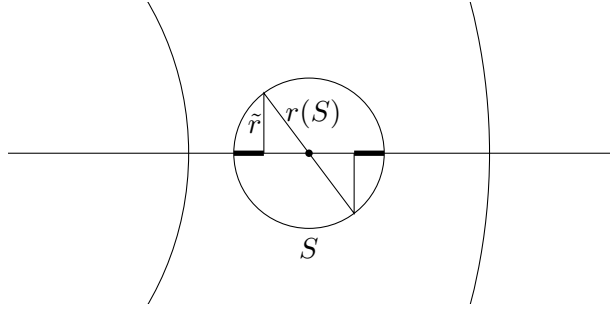

\vspace{5pt} \noindent \textbf{Step 4:} Define $\fS_3=\fS_2\setminus \fX_3$. We perform the following iterative process. Start with $\fX_4=\emptyset$ and $\fS_4=\fS_3$. We will define spheres $S_1,S_2,\ldots,S_\ell$. Suppose we have defined $S_1,\ldots, S_{i-1}$. Define $S_i$ to be a sphere in $\fS_4$ which intersects $S_{\mathrm{rand}}$ and has maximal radius among such spheres. Place the neighborhood $N(S_i)$ in $\fX_4$ and remove all spheres which intersect the ball bounded by $S_i$ from $\fS_4$ (in particular we remove $S_i$ and $N(S_i)$ from $\fS_4$). Repeat until no more spheres in $\fS_4$ intersect $S_{\mathrm{rand}}$.

We claim that the balls bounded by $S_1,S_2,\ldots,S_\ell$ are disjoint. For $i<j$ we know that $r(S_i)\geq r(S_j)$ and the ball bounded by $S_i$ is disjoint from the sphere $S_j$. The latter implies that $S_i,S_j$ do not intersect and that $S_i$ does not contain $S_j$; the former implies that $S_j$ does not contain $S_i$. 

\begin{claim}
\label{claim:x4-size}
$|\fX_4|\leq\Sigma/4$.
\end{claim}
\begin{proof}
By the definition of $\fX_3$, every sphere $S\in\fS_3$ which intersects $S_{\mathrm{rand}}$ also satisfies $\tilde r(S)>C_d'R(\deg(S)/\Sigma)^{1/(d-1)}$. Thus
\[|\fX_4|\leq\sum_{i=1}^\ell\deg(S_i)<\sum_{i=1}^\ell\Sigma\paren{\frac{\tilde r(S_i)}{C_d' R}}^{d-1}.\]
Now since the balls bounded by $S_1,\ldots,S_\ell$ are disjoint, we have
\[\sum_{i=1}^\ell \tau_{d-1}\tilde r(S_i)^{d-1}\leq \sigma_d (2R)^{d-1}\]
where $\tau_k,\sigma_k$ denote the $k$-volume of the unit ball in $\R^k$ and the $(k-1)$-volume of the unit sphere in $\R^k$, respectively. (In the above inequality, we use the fact that the $(d-1)$-volume of a spherical cap of radius $\tilde r$ on any sphere in $\R^d$ is at least the $(d-1)$-volume of a ball of radius $\tilde r$ in $\R^{d-1}$.) Thus we conclude that
\[|\fX_4|\leq \Sigma \frac{2^{d-1}\sigma_d}{C_d'^{d-1}\tau_{d-1}}\leq\Sigma/4,\]
where the last inequality holds for $C_d'$ chosen sufficiently large in terms of $d$.
\end{proof}

\begin{claim}
$\fX=\fX_1\cup \fX_2\cup \fX_3\cup \fX_4$ is a $(1-c_d)$-balanced separator for $G(\fS)$. 
\end{claim}
\begin{proof}
By Claim~\ref{claim:sphere-separator-proof}, it suffices to show that each connected component of $G(\fS_2)\setminus(\fX_3\cup \fX_4)$ has size at most $(1-c_d)n$. Note that any sphere in $\fS_2\setminus(\fX_3\cup \fX_4)$ which crosses $S_{\mathrm{rand}}$ must lie in the ball bounded by one of the $S_i$. Furthermore, no sphere in $\fS_2\setminus(\fX_3\cup \fX_4)$ crosses any of the $S_i$. Therefore any connected component of $G(\fS_2)\setminus(\fX_3\cup \fX_4)$ must either lie entirely within some sphere $S_i$, entirely within $S_{\mathrm{rand}}$, or entirely outside of $S_{\mathrm{rand}}$. We chose $\fS_2$ so that any sphere in $\fS_2$ contains at most $4c_dn$ centers of spheres in $\fS_2$. Thus the first type of connected component has size at most $4c_dn$. For the second type, note that $S_{\mathrm{rand}}$ is contained within $B_{\mathrm{out}}$ which we argued earlier contains at most $4^{d}\ceil{4c_dn}<(1-c_d)n$ centers of spheres in $\fS_2$. For the third type of connected component, we know that there are at least $4c_dn$ spheres which intersect $B_{\mathrm{in}}$. Since $B_{\mathrm{in}}$ is contained in the interior of $S_{\mathrm{rand}}$, we see that there are at most $(1-4c_d)n$ spheres outside of $S_{\mathrm{rand}}$, giving the desired bound on the third type of connected component.
\end{proof}
Since $|\fX|\leq\Sigma$, this completes the proof.
\end{proof}

Intersection graphs of spheres are an example of semialgebraic graphs. Using known bounds on the number of edges in $K_{t,t}$-free semialgebraic graphs \cite{TY24}, we deduce the following weak bound for separators in $K_{t,t}$-free sphere intersection graphs. 

\begin{corollary}
\label{cor:ktt-free-separator-weak}
Let $\fS$ be a collection of $n$ spheres in $\R^d$. If $G(\fS)$ is $K_{t,t}$-free, then it has a balanced separator of size $O_d(t^{2/d(d+2)}n^{1-2/d(d+2)})$.
\end{corollary}

\begin{proof}
Represent the sphere centered at $x\in\R^d$ with radius $r$ by the point $(x,r)\in\R^{d+1}$. Now two spheres $(x,r)$ and $(x',r')$ intersect if and only if $(r-r')^2\leq\|x-x'\|^2\leq(r+r')^2$. These are polynomial inequalities of degree $O(1)$, meaning that $G(\fS)$ can be represented by a semialgebraic graph in $\R^{d+1}$ of total degree $O(1)$. (See \cite[Definition 1.1]{TY24} for the full definition.)

By \cite[Theorem 1.6]{TY24} (a bound on the Zarankiewicz problem for semialgebraic graphs), if $G(\fS)$ is $K_{t,t}$-free, then it has $O_d(t^{2/(d+2)}n^{2(d+1)/(d+2)})$ edges. Applying Theorem~\ref{thm:sphere-separator-main-intro} with this bound shows that $G(\fS)$ has a balanced separator of size
\[O_d\paren{\paren{t^{2/(d+2)}n^{2(d+1)/(d+2)}}^{1/d}n^{1-2/d}}=O_d(t^{2/d(d+2)}n^{1-2/d(d+2)}).\qedhere\]
\end{proof}

We can bootstrap this weak bound to get a better bound on the number of edges in a $K_{t,t}$-free sphere intersection graph. This follows from \cite[Lemma 8]{FP08} which turns a bound on the size of separators into a bound on the number of edges for any hereditary family of graphs. (See \cite[Section 3]{FP14} for another application of this method.)

\begin{corollary}
\label{cor:edge-bound-ktt-free}
Let $\fS$ be a collection of $n$ spheres in $\R^d$. If $G(\fS)$ is $K_{t,t}$-free, then it has $O_d(tn)$ edges.
\end{corollary}

\begin{proof}
We aim to apply \cite[Lemma 8]{FP08}. This result requires a monotone function $\phi$ so that any $n$-vertex $K_{t,t}$-free intersection graph of spheres in $\R^d$ contains a balanced separator of size at most $n\phi(n)$. By Corollary~\ref{cor:ktt-free-separator-weak}, we can take $\phi(n)=C_d(t/n)^{2/d(d+2)}$ where $C_d$ is the constant in that result. Clearly $\phi(n)$ is monotone.

Set $n_0=(12C_d)^{d(d+2)/2}t$ so that $\phi(n_0)=1/12$. Then \cite[Lemma 8]{FP08} implies that every $K_{t,t}$-free intersection graph of spheres in $\R^d$ on $n\geq n_0$ vertices has at most $Cn_0n/2$ edges where
\[C=\prod_{i=0}^\infty \paren{1+\phi\paren{\ceil{(4/3)^in_0}}}.\]
We can bound
\[C\leq\exp\paren{\sum_{i=0}^\infty \phi\paren{\ceil{(4/3)^in_0}}}\leq\exp\paren{2C_d\sum_{i=0}^\infty \paren{\paren{\frac34}^i\frac{t}{n_0}}^{2/d(d+2)}}.\]
As $n_0=O_d(t)$, this geometric series sums to $O_d(1)$, showing that $C=O_d(1)$. This gives the desired bound on the number of edges when $n\geq n_0$.

Note that if $n<n_0$, then we can trivially bound the number of edges by $n^2<n_0n=O_d(tn)$.
\end{proof}

Combining Theorem~\ref{thm:sphere-separator-main-intro} and Corollary~\ref{cor:edge-bound-ktt-free}, we immediately conclude the following strengthening of Corollary~\ref{cor:ktt-free-separator-weak}. This result improves upon the sphere separator result of Davies--Georgakopoulos--Hatzel--McCarty \cite[Theorem 1]{DGHM25}, proving a strengthening of \cite[Conjecture 20]{DGHM25}.

\begin{corollary}
\label{cor:ktt-free-separator}
Let $\fS$ be a collection of $n$ spheres in $\R^d$. If $G(\fS)$ is $K_{t,t}$-free, then it has a balanced separator of size $O_d(t^{1/d}n^{1-1/d})$.
\end{corollary}

\begin{remark}
The same proof also gives Corollaries~\ref{cor:edge-bound-ktt-free} and \ref{cor:ktt-free-separator} for intersection graphs of balls. However these results actually have significantly easier proofs. For example, let $\fB$ be a collection of balls in $\R^d$. If $\fB$ is $(2t-1)$-ply, then Theorem~\ref{thm:ball-separator-ply} already implies that $G(\fB)$ has a balanced separator of size $O_d(t^{1/d}n^{1-1/d})$. Otherwise, there is a point in $2t$ balls of $\fB$, implying that $G(\fB)$ contains $K_{2t}$ which contains $K_{t,t}$.
\end{remark}

\section{Further results}
\label{sec:further}

We start by giving a construction that shows that our main results are asymptotically optimal. We then generalize these results to a larger class of geometric objects.

\begin{proposition}
\label{prop:lower-bound}
For each $n,d\geq 1$ and $n\leq m\leq \binom n2$, there exists an intersection graph of $n$ balls or spheres in $\R^d$ with at most $m$ edges for which every balanced separator has size at least $\Omega_d(m^{1/d}n^{1-2/d})$.
\end{proposition}

\begin{proof}
Let $V=\{1,\ldots,k\}^d$ for $k=\floor{n^{1/d}}$. Define $G=(V,E)$ where $x,y\in V$ are adjacent if $\|x-y\|\leq r$. Choosing $1\leq r=\Theta_d((m/n)^{1/d})$ appropriately, we can ensure that $G$ has $\Theta_d(m)$ edges. Note that $G$ is both the intersection graph of the radius $r/2$ balls centered at the points of $V$ and the intersection graph of the radius $r/2$ spheres centered at the points of $V$. We will show that every balanced separator of $G$ has size at least $\Omega_d(m^{1/d}n^{1-2/d})$. 

We first handle the case when $r=1$. Suppose that $\fX$ is a balanced separator with $|\fX|\leq k^d/3$. Let $A\subset V$ be the union of some connected components of $G\setminus\fX$, satisfying $k^d/3\leq|A|\leq 2k^d/3$. When $r=1$, the graph $G$ is known as the grid graph on $\{1,\ldots,k\}^d$. Then the isoperimetric inequality on the grid graph \cite{BL91} implies that there are at least $k^{d-1}$ edges between $A$ and $V\setminus A$. However, all such edges must have one endpoint in $\fX$. Since the maximum degree of $G$ is $2d$, we conclude that $|\fX|\geq k^{d-1}/2d\geq\Omega_d(n^{1-1/d})$, as desired.

To handle the general case, let $r'=r/(2\sqrt{d})$. For simplicity, assume that $r'$ is an integer that divides $k$, say $k=r'k'$. Divide $\{1,\ldots,k\}^d$ into $(k')^d$ subcubes of side length $r'$. Note that every vertex in one subcube is adjacent in $G$ to every vertex in an adjacent subcube. Similar to before, suppose that $\fX$ is a balanced separator with $|\fX|\leq k^d/12$ and $A$ is a union of connected components of $G\setminus \fX$ with $k^d/3\leq |A|\leq 2k^d/3$. Write $B=V\setminus(\fX\cup A)$. We know $|B|\geq k^d/4$. Define $A'\subseteq\{1,\ldots,k'\}^d$ where a point lies in $A'$ if the corresponding $r'\times\cdots\times r'$ subcube of $\{1,\ldots,k\}^d$ contains at least one point of $A$. Define $B'\subseteq\{1,\ldots,k'\}^d$ similarly. First note that $A'$ and $B'$ are disjoint -- if they shared a vertex then $A,B$ would contain points in the same subcube which have distance at most $\sqrt{d}r'<r$.

Write $X'=\{1,\ldots,k'\}^d\setminus(A'\cup B')$. Note that $|A'|\geq|A|/(r')^d\geq (k')^d/3$ and similarly $|B'|\geq|B|/(r')^d\geq (k')^d/4$. Thus $|A'|\leq (k')^d-|B'|\leq 3(k')^d/4$. Applying the isoperimetric inequality to the grid graph on $\{1,\ldots,k'\}^d$, we conclude that there are at least $(k')^{d-1}$ edges between $A'$ and $B'\cup X'$. However, there are no edges between $A',B'$ because any two points in two adjacent subcube are at distance at most $2\sqrt{d}r'=r$. Thus we conclude that $|X'|\geq(k')^{d-1}/(2d)$. Note though that a point lies in $X'$ only if the corresponding subcube lies entirely in $\fX$. We thus see that $|\fX|\geq(r')^d(k')^{d-1}/(2d)=\Theta_d(m^{1/d}n^{1-2/d})$.
\end{proof}

Our main results on ball and sphere separators, Theorems~\ref{thm:ball-separator-lp}~and~\ref{thm:sphere-separator-main}, hold for more general classes of geometric objects. We say that a set $B\subset\R^d$ is \emph{$C$-fat} if $B$ contains a ball of radius $r$ and is contained inside a ball of radius $Cr$ for some $r$. Our results hold for fat convex bodies and for the boundaries of fat convex bodies. (See \cite{HQ17,BBKMZ20} for related results under slightly different hypotheses.)

\begin{theorem}
\label{thm:ball-separator-fat-convex}
Let $\fB$ be a collection of $n$ $C$-fat convex bodies in $\R^d$. Then $G(\fB)$ has a balanced separator of size at most
\[O_{d,C}\paren{\paren{\sum_{B\in\fB}p(B)^{\tfrac1{d-1}}}^{1-\tfrac1d}}.\]
In particular, if $G(\fB)$ has $m$ edges, then it has a balanced separator of size $O_{d}(m^{1/d}n^{1-2/d})$.
\end{theorem}

\begin{proof}
For each $B\in\fB$, pick an arbitrary point $p\in\inter(B)$ that we call the \emph{center} of $B$. With this definition, we proceed with the proof of Theorem~\ref{thm:ball-separator-lp} in the same way. The same randomized procedure produces a $(1-5^{-d})$-balanced separator.

For each $B\in\fB$, let $x_B\subset\R_{\geqslant0}$ be the set of radii $r$ so that the origin-centered sphere of radius $r$ intersects $B$. We wish to show that if $x_B\cap[1,2]$ is non-empty, $B$ contains a genuine ball $B'$ which is contained in $B(2)$ such that $r(B')\geq\Omega_C\paren{|x_B\cap[1,2]|}$.

First, if $x_B\subseteq[0,2]$, then $B$ is contained in $B(2)$. Now $B$ contains a point of norm $\max(x_B)$ and a point of norm $\min(x_B)$, i.e., it contains two points at distance at least $\max(x_B)-\min(x_B)\geq |x_B|$. This means that any ball containing $B$ must have radius at least $|x_B|/2$. By hypothesis, $B$ contains a ball of radius $|x_B|/2C$. Letting $B'$ be this ball, we are done since in this case $B'\subseteq B$ is contained in $B(2)$.

Now suppose $x_B\not\subseteq[0,2]$, but $x_B\cap [0,2]\neq\emptyset$. By the same argument, $B$ contains some ball (say $\tilde B$) of radius $(\max(x_B)-\min(x_B))/2C$ centered at some point $p$ with $\|p\|\leq \max(x_B)$. By assumption, $B$ contains a point $q$ with $\|q\|=\min(x_B)\leq 2$. (See Figure~\ref{fig:expected-S-size-fat-convex}.) Then the convexity of $B$ implies that $B$ contains the convex hull of $\{q\}\cup\tilde B$. Elementary geometry implies that this convex hull contains a ball $B'$ of radius $\tfrac{2-\min(x_B)}{2C+1}\geq \tfrac{|x_B\cap[1,2]|}{2C+1}$ that is contained in $B(2)$. (To see this, dilate the convex hull of $\{q\}\cup\tilde B$ by a factor of $(2-\|q\|)/(\|p\|+r(\tilde B)-\|q\|)$ so that it lies within $B(2)$.) 

Thus we have constructed a set $\fB'$ contained in $B(2)$ so that $\E[|\fX|]\geq\Omega_C\paren{\sum_{B\in\fB'}r(B)}$. Note that $\fB'$ consists of genuine balls, so the rest of the proof applies unchanged.
\end{proof}

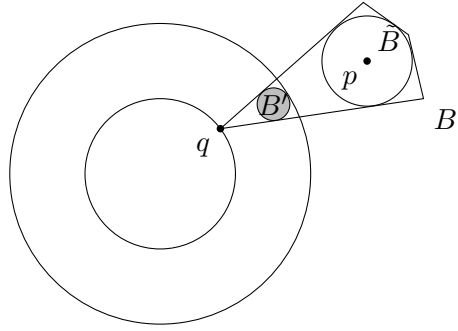
\begin{figure}
\centering
\begin{tikzpicture}

\draw(0,0) circle (1);
\draw(0,0) circle (2);

\coordinate (Q) at (0.8, 0.6);
\coordinate (C) at (3.3, 1.85); 
\coordinate (A) at (2.7, 2.28);
\coordinate (B) at (3.5, 1); 
\coordinate (P) at (2.75,1.5);
\coordinate (R) at ({0.8*0.638+2.75*0.362},{0.6*0.638+1.5*0.362});

\draw(P) circle (0.6);
\draw[black] (Q) -- (A) -- (C) -- (B) -- cycle;

\draw[black, fill=lightgray] (R) circle ({0.362*0.6});
\node at (R) {$B'$};

\node[] at (P) [above right] {$\tilde B$};
\node[] at (B) [below right] {$B$};

\node[fill=black, circle, inner sep=1pt] at (Q) {};
\node[] at (Q) [below left] {$q$};

\node[fill=black, circle, inner sep=1pt] at (P) {};
\node[] at (P) [below left] {$p$};

\end{tikzpicture}
\caption{\label{fig:expected-S-size-fat-convex}
The construction of $B'$ inside of $B\cap B(2)$.}
\end{figure}

\begin{theorem}
\label{thm:sphere-separator-fat-convex}
Let $\fS$ be a collection of $n$ boundaries of $C$-fat convex bodies in $\R^d$. Then $G(\fS)$ has a balanced separator of size at most
\[O_{d,C}\paren{\paren{\sum_{S\in\fS}\deg(S)^{\tfrac1{d-1}}}^{1-\tfrac1d}}.\]
In particular, if $G(\fS)$ has $m$ edges, then it has a balanced separator of size $O_{d}(m^{1/d}n^{1-2/d})$.
\end{theorem}

\begin{proof}
For each $S\in\fS$, pick an arbitrary point in the interior of $S$ to be its center. Consider the partial order on $\fS$ where $S\prec S'$ if $S$ is contained in the interior of $S'$. Pick an arbitrary linear extension of this partial order. When we talk about the $S\in\fS$ of maximal radius (or minimal radius) we mean maximal with respect to this linear extension. With these two definitions, we can proceed with the proof of Theorem~\ref{thm:sphere-separator-main} in the same way. The same randomized procedure produces a $(1-c_d)$-balanced separator.

The main difference in the analysis is understanding $\E[|\fX_3|]$. For each $S$, define $\tilde r(S)$ so that the $(d-1)$-volume of the intersection of the convex body bounded by $S$ and the sphere $S_{\mathrm{rand}}$ is $\tau_{d-1}\tilde r(S)^{d-1}$. We claim that $\Pr[\tilde r(S)\leq \tilde r]\leq O_C(\tilde r/R)$. 

Fix $\tilde r$. Change coordinates so that $S_{\mathrm{rand}}$ is centered at the origin. Write $B$ for the convex body with boundary $S$.  Let $r_{\min}=\min_{x\in B}\|x\|$ and $r_{\max}=\max_{x\in B}\|x\|$. Pick $p,q\in B$ with $\|p\|=r_{\min}$ and $\|q\|=r_{\max}$. Since $\|p-q\|\geq r_{\max}-r_{\min}$, we see that the smallest ball that contains $B$ has radius at least $(r_{\max}-r_{\min})/2$. By the $C$-fat hypothesis, $B$ contains some ball $B'$ centered at a point $o$ with radius $(r_{\max}-r_{\min})/2C$. Now $B$ contains the convex hull of $\{p,q\}\cup B'$ which contains two cones with apexes $p,q$ and with bases some disk of radius $(r_{\max}-r_{\min})/2C$ centered at $o$. This implies that for each $\epsilon>0$, the body $B$ contains a ball of radius $\epsilon(r_{\max}-r_{\min})/2C$ centered at a point of distance $r_{\min}+\epsilon(r_{\max}-r_{\min})$ from the origin and another ball of the same radius centered at a point of distance $r_{\max}-\epsilon(r_{\max}-r_{\min})$ from the origin. (See Figure~\ref{fig:expected-x3-size-fat-convex-boundary}.) In particular, if $\epsilon(r_{\max}-r_{\min})/2C\geq \tilde r$, then the intersection of $B$ and the spheres of radius $r_{\min}+\epsilon(r_{\max}-r_{\min})$ and $r_{\max}-\epsilon(r_{\max}-r_{\min})$ both contain a spherical cap of radius at least $\epsilon(r_{\max}-r_{\min})/2C\geq \tilde r$, so these intersections have $(d-1)$-volume at least $\tau_{d-1}\tilde r^{d-1}$. The first inequality holds for $\epsilon\geq 2C\tilde r/(r_{\max}-r_{\min})$, so
\[\Pr[\tilde r(S)\leq \tilde r]\leq\frac{2(r_{\max}-r_{\min})\cdot 2C\tilde r/(r_{\max}-r_{\min})}{R}=\frac{4C\tilde r}{R},\]
as desired.

This gives us the bound $\E[|\fX_3|]\leq \Sigma/4$ as in Claim~\ref{claim:expected-x3-size}. Then the bound $|\fX_4|\leq \Sigma/4$ also holds since we defined $\tilde r(S)$ so that the proof of Claim~\ref{claim:x4-size} still goes through.
\end{proof}

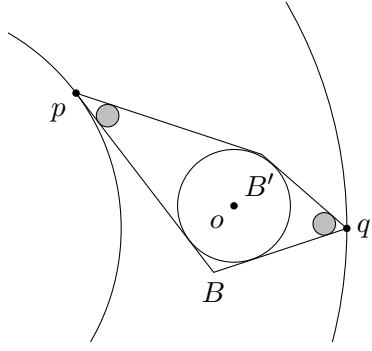
\begin{figure}
\centering
\begin{tikzpicture}[scale=3]
\clip (0.5,-0.5) rectangle + (2,1.5);

\draw(0,0) circle (1);
\draw(0,0) circle (2);

\coordinate (P) at (0.8, 0.6);
\coordinate (Q) at (2, 0); 
\coordinate (A) at (1.623, 0.328);
\coordinate (B) at (1.41, -0.195); 
\coordinate (O) at (1.5, 0.1);
\coordinate (M) at ({0.8*0.8+1.5*0.2},{0.6*0.8+0.1*0.2});
\coordinate (N) at ({2*0.8+1.5*0.2},{0*0.8+0.1*0.2});

\draw(O) circle (0.25);
\draw[black, fill=lightgray](M) circle (0.05);
\draw[black, fill=lightgray](N) circle (0.05);
\draw[black] (P) -- (A) -- (Q) -- (B) -- cycle;

\node[fill=black, circle, inner sep=1pt] at (O) {};
\node[] at (O) [below left] {$o$};
\node[] at (O) [above right] {$B'$};

\node[fill=black, circle, inner sep=1pt] at (Q) {};
\node[] at (Q) [right] {$q$};

\node[fill=black, circle, inner sep=1pt] at (P) {};
\node[] at (P) [below left] {$p$};

\node[] at (B) [below] {$B$};

\end{tikzpicture}
\caption{\label{fig:expected-x3-size-fat-convex-boundary}
Bounding $\Pr[\tilde r(S)\leq \tilde r]$.}
\end{figure}

With these results, Corollary~\ref{cor:edge-bound-ktt-free} also holds for fat convex bodies if the bodies are also semialgebraic sets of bounded description complexity.

We also note that both of the assumptions, convexity and fatness, are necessary. Indeed, for $d\geq 3$, every graph is the intersection graph of convex bodies in $\R^d$ and also every graph is the intersection graph of fat connected bodies in $\R^d$. The former fact is classical (see \cite{Tie05}); we prove the latter fact here.

\begin{proposition}
\label{prop:any-graph-fat-non-convex}
For $d\geq 3$, $\epsilon>0$, and any graph $G$, there exists a collection of $(1+\epsilon)$-fat connected bodies in $\R^d$ whose intersection graph is $G$.
\end{proposition}

\begin{proof}
Let $B_k$ be the ball centered at the origin in $\R^d$ of radius $r_k=(2+2\epsilon^{-1})^k$. We will define $X_k$ to be a $(1+\epsilon)$-fat connected body that is contained in $B_k$. The set $X_k$ will be a ``balloon'' consisting of a large sphere in $\inter(B_k)\setminus\inter( B_{k-1})$ that is tangent to $B_{k-1}$ at a point $p_k$ and a ``string'' in $B_{k-1}$: a curve starting at $p_k\in\partial B_{k-1}$, passing through $\partial B_0$ at a point $q_k$, and then continuing into $B_0$. We will define these sets $X_k$ so that they are disjoint in $\R^d\setminus\inter (B_0)$. See Figure~\ref{fig:any-graph-fat-non-convex} for an illustration of this construction.

Suppose we have defined $B_0,\ldots,B_{k-1}$ and $X_1,\ldots,X_{k-1}$. Let $p_k$ be a point on $\partial B_{k-1}$. Say $B_{k-1}$ has radius $r_{k-1}$. Let $Y_k$ be the ball of radius $\epsilon^{-1}r_{k-1}$, centered at $(1+\epsilon^{-1})p_k$. Then $B_{k-1}\cup Y_k$ is contained in the ball of radius $(1+\epsilon)\epsilon^{-1}r_{k-1}$ centered at $\epsilon^{-1}p_k$. Since $r_k=2(1+\epsilon^{-1})r_{k-1}$, we see that $Y_k$ is contained in $\inter (B_k)$. We will define $X_k$ to be $Y_k\cup\gamma_k$ where $\gamma_k$ is a curve starting from $p_k$ which is contained in $B_{k-1}$. Whatever our choice of $\gamma_k$, we have ensured that $X_k$ is $(1+\epsilon)$-fat and connected.

Pick $q_k$ to be a point in $\partial B_0\setminus \{q_1,\ldots, q_{k-1}\}$. Let $\gamma'_k$ be a curve from $p_k$ to $q_k$ whose interior is contained in
\[\inter(B_{k-1})\setminus\paren{B_0\cup\bigcup_{i=1}^{k-1}X_i}.\]
If $k=1$, we set $q_1=p_1$ and let $\gamma'_1$ have length 0. For $k>1$, clearly there is enough space in dimension $d\geq 3$ to find such a $\gamma'_k$.

Finally, define $\gamma_k=\gamma'_k\cup\delta_k$ where $\delta_k$ is a curve starting from $q_k$ whose interior is contained in $\inter(B_0)$ such that $\delta_k$ and $\delta_\ell$ intersect if and only if $k,\ell$ are adjacent in $G$. Again, in dimension $d\geq 3$, this is clearly possible.
\end{proof}

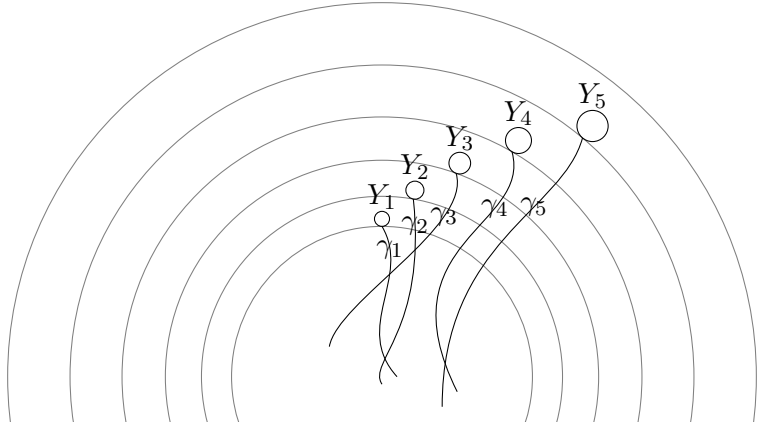
\begin{figure}
\centering
\begin{tikzpicture}[scale=2]
\clip (-2.5,-0.3) rectangle + (5,2.8);

\draw[gray] circle (1);
\draw[gray] circle (1.2);
\draw[gray] circle (1.44);
\draw[gray] circle (1.728);
\draw[gray] circle (2.0736);
\draw[gray] circle (2.48832);

\draw(0,1.05) circle (0.05);
\draw(1.05*1.2*0.1736,1.05*1.2*0.9848) circle (0.05*1.2);
\draw(1.05*1.44*0.342,1.05*1.44*0.9397) circle (0.05*1.44);
\draw(1.05*1.728*0.5,1.05*1.728*0.866) circle (0.05*1.728);
\draw(1.05*2.0736*0.6428,1.05*2.0736*0.766) circle (0.05*2.0736);

\node at (0,1.05) [above] {$Y_1$};
\node at (1.05*1.2*0.1736,1.05*1.2*0.9848+0.01) [above] {$Y_2$};
\node at (1.05*1.44*0.342,1.05*1.44*0.9397+0.02) [above] {$Y_3$};
\node at (1.05*1.728*0.5,1.05*1.728*0.866+0.033) [above] {$Y_4$};
\node at (1.05*2.0736*0.6428,1.05*2.0736*0.766+0.055) [above] {$Y_5$};

\draw(0,1) .. controls (0.2,0.66) and (-0.2,0.33) .. (0.1,0) node[pos=0.15] {$\gamma_1$};
\draw(1.2*0.1736,1.2*0.9848) .. controls (0.3,0.3) and (-0.1,0.1) .. (0,-0.05) node[pos=0.07] {$\gamma_2$};
\draw(1.44*0.342,1.44*0.9397) .. controls (0.6,1) and (-0.3,0.5) .. (-0.35,0.2) node[pos=0.25] {$\gamma_3$};
\draw(1.728*0.5,1.728*0.866) .. controls (1,1) and (0,0.9) .. (0.5,-0.1) node[pos=0.3] {$\gamma_4$};
\draw(2.0736*0.6428,2.0736*0.766) .. controls (1.2,1.1) and (0.4, 0.9) .. (0.4,-0.2) node[pos=0.35] {$\gamma_5$};

\end{tikzpicture}
\caption{\label{fig:any-graph-fat-non-convex}
The construction in Proposition~\ref{prop:any-graph-fat-non-convex}.}
\end{figure}

\section{Algorithmic implementations}
\label{sec:alg}

Suppose we are given $\fB$ as a collection of centers and radii of balls. We give a randomized linear time algorithm that computes a small separator for $G(\fB)$. 

\begin{theorem}
\label{thm:ball-separator-alg}
The separator for intersection graphs of balls given in Theorem~\ref{thm:ball-separator-lp} can be computed by a randomized linear time algorithm that succeeds with probability at least $1/2$.
\end{theorem}

\begin{proof}
Write $\cP$ for the centers of the balls of $\fB$. Let $r$ be the smallest radius such that there exists a ball $B_0$ of radius $r$ that contains at least $9^{-d}n$ points of $\cP$. Fix such a ball $B_0$. Then for any $p\in \cP\cap B_0$, the ball of radius $2r$ centered at $p$ contains $B_0$ and thus contains at least $9^{-d}n$ points of $\cP$.

Our first goal is to efficiently find an approximation for this ball $B_0$. Fix a large constant $C_d$ and sample $C_d\cdot 9^{d}$ points $p_1,p_2,\ldots\in\cP$ uniformly at random. For each point $p_i$, compute the smallest $r_i$ such that the ball of radius $r_i$ centered at $p_i$ contains $9^{-d}n$ points of $\cP$. This can be computed in linear time for each $p_i$. Say we succeed if $\min_i r_i\leq 2r$. By the discussion above, we succeed if some $p_i\in\cP\cap B_0$, which occurs with probability at least
\[1-\paren{1-9^{-d}}^{C_d\cdot 9^d}\geq 1-e^{-C_d}.\]

Pick $1\leq i\leq C_d\cdot9^d$ such that $r_i$ is minimal. As in the proof of Theorem~\ref{thm:ball-separator-lp}, pick $r_{\mathrm{rand}}\in[r_i,2r_i]$ uniformly at random and let $\fX$ be the set of balls $B\in\fB$ which intersect the sphere of radius $r_{\mathrm{rand}}$ centered at $p_i$. This set can be computed in linear time. Now we still have \[\E[|\fX|]\leq 4\paren{\sum_{B\in\fB}p(B)^{\tfrac1{d-1}}}^{1-\tfrac1d},\]
so by Markov's inequality,
\[\Pr\sqb{|\fX|\leq 4C_d\paren{\sum_{B\in\fB}p(B)^{\tfrac1{d-1}}}^{1-\tfrac1d}}\geq 1-C_d^{-1}.\]

Now the set $\fX$ is a $(1-9^{-d})$-balanced separator for $G(\fB)$. To see this, note that the ball of radius $r_i$ centered at $p_i$ contains $B_0$ which contains $\ceil{9^{-d}n}$ points of $\cP$. Furthermore, the ball of radius $2r_i$ centered at $p_i$ can be covered by at most $8^d$ translates of $B_0$, so it contains at most $8^d\ceil{9^{-d}}n<(1-9^{-d})n$ points of $\cP$.

Thus we have a randomized linear time algorithm which computes a small $(1-9^{-d})$-balanced separator for $G(\fB)$ with probability at least $1-C_d^{-1}$. By Lemma~\ref{lem:balancing-separators}, for $C_d$ sufficiently large, we can iterate this algorithm $\ceil{\log_{1-9^{-d}}(2/3)}$ times to produce a $\tfrac23$-balanced separator for $G(\fB)$ still in linear time with success probability at least $\tfrac12$.
\end{proof}

Suppose we are given $\fS$ as a list of centers and radii of the spheres. In quadratic time we can compute which pairs of spheres intersect and which pairs are contained in each other. From this we can run the algorithm described in Theorem~\ref{thm:sphere-separator-main} in quadratic time. The only potential difficulty is in selecting the minimal radius ball $B_{\mathrm{in}}$, but this can be handled the same way is in Theorem~\ref{thm:ball-separator-alg}.

\begin{theorem}
\label{thm:sphere-separator-alg}
The separator for intersection graphs of spheres given in Theorem~\ref{thm:sphere-separator-main} can be computed by a randomized quadratic time algorithm that succeeds with probability at least $1/2$.
\end{theorem}


\end{document}